\pdfoutput=1

\documentclass{article}
\PassOptionsToPackage{numbers, compress}{natbib}

\usepackage{fullpage}
\usepackage{url,color,cite}
\usepackage{caption}
\usepackage{amsmath,amsthm,amsfonts,amssymb}
\usepackage{mathtools}

\usepackage{float}
\usepackage{epstopdf}
\usepackage{nicefrac}
\usepackage{algorithm}
\usepackage[noend]{algpseudocode}
\usepackage{pifont}
\usepackage{subcaption}
\usepackage{sidecap}
\usepackage{multirow}
\usepackage{multicol}
\usepackage{bbm}
\usepackage{enumitem}

\newtheorem{lemma}{Lemma}
\newtheorem*{lemma*}{Lemma}
\newtheorem{theorem}{Theorem}
\newtheorem*{theorem*}{Theorem}
\theoremstyle{definition}
\newtheorem{definition}{Definition}
\newtheorem*{definition*}{Definition}
\theoremstyle{definition}

\theoremstyle{definition}
\newtheorem{corollary}{Corollary}
\newtheorem*{corollary*}{Corollary}
\theoremstyle{definition}

\newtheorem*{claim*}{Claim}
\newtheorem{conjecture}{Conjecture}

\newcommand{\calA}{\mathcal{A}}

\newcommand{\calD}{\mathcal{D}}

\newcommand{\calM}{\mathcal{M}}

\newcommand{\calS}{\mathcal{S}}

\newcommand{\eps}{\epsilon}

\usepackage{microtype}
\usepackage{graphicx}
\usepackage{booktabs} 
\usepackage{hyperref}
\usepackage{cleveref}

\usepackage{lipsum}

\usepackage[preprint]{neurips_2022}

\usepackage[utf8]{inputenc} %
\usepackage[T1]{fontenc}    %
\usepackage{hyperref}       %
\usepackage{url}            %
\usepackage{booktabs}       %
\usepackage{amsfonts}       %
\usepackage{nicefrac}       %
\usepackage{microtype}      %
\usepackage{xcolor}         %
\usepackage{wrapfig}

\title{Shuffle Gaussian Mechanism for  Differential Privacy}

\author{Seng Pei Liew\thanks{LINE Corporation; \texttt{\{sengpei.liew,tsubasa.takahashi\}@linecorp.com}}
\And
Tsubasa Takahashi
}

\date{}

\begin{document}

\maketitle

\begin{abstract}
We study Gaussian mechanism in the shuffle model of differential privacy (DP).
Particularly, we characterize the mechanism's R{\'e}nyi differential privacy (RDP), showing that it is of the form:
$$
\eps(\lambda) \leq \frac{1}{\lambda-1}\log\left(\frac{e^{-\lambda/2\sigma^2}}{n^\lambda}\sum_{\substack{k_1+\dotsc+k_n=\lambda;\\k_1,\dotsc,k_n\geq 0}}\binom{\lambda}{k_1,\dotsc,k_n}e^{\sum_{i=1}^nk_i^2/2\sigma^2}\right)
$$
We further prove that the RDP is strictly upper-bounded by the Gaussian RDP without shuffling.
The shuffle Gaussian RDP is advantageous in composing multiple DP mechanisms, where we demonstrate its improvement over the state-of-the-art approximate DP composition theorems in privacy guarantees of the shuffle model.
Moreover, we extend our study to the subsampled shuffle mechanism and the recently proposed shuffled check-in mechanism, which are protocols geared towards distributed/federated learning. 
 Finally, an empirical study of these mechanisms is given to demonstrate the efficacy of employing shuffle Gaussian mechanism under the distributed learning framework to guarantee rigorous user privacy.  
 \footnote{The source code of our implementation is available at \url{https://github.com/spliew/shuffgauss}.}
\end{abstract}

{\allowdisplaybreaks

\section{Introduction}\label{sec:introduction}

The shuffle/shuffled model \cite{cheu2019distributed,erlingsson2019amplification} has attracted attention recently as an intermediate model of trust in differential privacy (DP) \cite{dwork2006calibrating,dwork2006our}.
Within this setup, each user sends a locally differentially private (LDP) \cite{kasiviswanathan2011can} report to a trusted shuffler, where the collected reports are anonymized/shuffled, before being forwarded to the untrusted analyzer/server.
When viewed in the central model, the central DP parameter, $\eps$, can be smaller than the local one, $\eps_0$ \cite{erlingsson2019amplification}; a phenomenon also known as privacy amplification.
This yields better utility-privacy trade-offs than the local model of DP in general, without relying on a highly trusted server as in the central model of DP.
The shuffle model can be realized in practice through a Trusted Execution Environment (TEE) \cite{bittau2017prochlo}, mix-nets \cite{chaum1981untraceable}, or peer-to-peer protocols \cite{liew2022network}.
Various aspects of the shuffle model have been investigated in the literature \cite{cheu2019distributed,balle2019privacy,balle2020privacy,cheu2021differential,girgis2021differentially,girgis2021renyia,girgis2021renyib,koskela2021tight,feldman2022hiding,imola2022differentially}.

A prominent use case of the shuffle model is in differentially private distributed learning, or federated learning (FL) \cite{mcmahan2017communication,kairouz2021advances}.
The non-private version of FL proceeds roughly as follows.
The ochestrating server initiates the learning by sending a model to users.
Then, each user calculates the gradient/model update using her own data, and sends the gradient to the server while keeping her data local.

In order to incorporate DP to distributed learning, it is convenient to first consider the probably most popular (and state-of-the-art) approach of learning with DP under the centralized setting, i.e.,  differentially private stochastic gradient descent (DP-SGD) \cite{bassily2014private,song2013stochastic,abadi2016deep}.
Here, the Gaussian mechanism is employed, where noise in the form of Gaussian/normal distribution is applied to a batch of clipped gradients to attain DP model updates.
As learning often requires repeated interaction, privacy composition is vital when working under a pre-determined privacy budget.
In DP-SGD, privacy is conventionally accounted for via the moments accountant \cite{abadi2016deep}, which is essentially an accounting method based on the more general R{\'e}nyi differential privacy (RDP) \cite{mironov2017renyi}, a notion of DP that provides generally much tighter DP composition, leading to a significant saving of privacy budget compared to strong composition theorems \cite{dwork2010boosting,kairouz2015composition}.
The subsampling effect \cite{ullman} is typically leveraged as well to achieve acceptable levels of privacy.
Moreover, the central-DP version of DP-SGD has been adapted for federated learning \cite{mcmahan2017learning}.

However, perhaps due to the difficulty of handling approximate DP mechanisms (particularly in privacy accounting), such as the Gaussian mechanism, the \textit{shuffle Gaussian mechanism}, and its applications in distributed learning have not been explored in depth \cite{koskela2021tight, liew2022shuffled}. 
Distributed learning in the local or shuffle model is often conducted by utilizing variants of the LDP-SGD algorithm \cite{duchi2018minimax,erlingsson2020encode} in the literature, which is arguably more sophisticated than the Gaussian mechanism. \footnote{The LDP-SGD algorithm proceeds roughly as follows. 
The client-side algorithm first clips user gradient, and flips its sign with a certain probability.
Then, a unit vector is sampled randomly to make an inner product with the processed gradient to yied the inner product's sign.
The sign is again flipped with a certain probability, before sending it along with the unit vector to the server.
The server-side algorithm includes normalizing the aggregated reports, and applying a $l_2$ projection when updating the model.
See Algorithms 4 and 5 in \cite{erlingsson2020encode} for details.}

In this paper, we provide a characterization of the RDP of the shuffle Gaussian mechanism.
Our principal result is showing that the R{\'e}nyi divergence of the shuffle Gaussian mechanism is
\begin{equation}\label{eq:main}
    D_{\lambda}(\calM(D)||\calM(D')) = \frac{1}{\lambda-1}\log\left(\frac{e^{-\lambda/2\sigma^2}}{n^\lambda}\sum_{\substack{k_1+\dotsc+k_n=\lambda;\\k_1,\dotsc,k_n\geq 0}}\binom{\lambda}{k_1,\dotsc,k_n}e^{\sum_{i=1}^nk_i^2/2\sigma^2}\right)
\end{equation}
To calculate the above expression numerically, we reduce it to a partition problem in number theory.
This enables us to perform tight privacy composition of any Gaussian-noise randomized mechanism coupled with a shuffler.
Furthermore, the shuffle Gaussian RDP characterization allows us to compute the RDP of subsampled shuffle Gaussian and shuffled check-in Gaussian mechanisms, which are mechanisms tailored to distributed learning.
To gauge the utility-privacy trade-offs of the variants of shuffle Gaussian mechanism in distributed learning, we also perform machine learning tasks experimentally in later sections.

\noindent\textbf{Related work.}
Early studies on the shuffle model put focus on $\eps_0$-local randomizer \cite{cheu2019distributed,erlingsson2019amplification,balle2019privacy}. In
\cite{balle2020privacy,feldman2022hiding}, extension to approximate DP, i.e., $(\eps_0,\delta_0)$-local randomizer has been worked on.
While methods for composing optimally approximate DP mechanisms are known \cite{kairouz2015composition,murtagh2016complexity}, they are unable to compose optimally mechanisms that satisfy multiple values of $(\eps,\delta)$ simultaneously, such as the Gaussian mechanism.

A preliminary analysis of the shuffle Gaussian mechanism using the Fourier/numerical accountant \cite{koskela2020computing,gopi2021numerical,zhu2022optimal} has been performed in \cite{koskela2021tight}.
This approach is known to give tight composition in general, but faces the curse of dimensionality;
only $n\lesssim 10$ ($n$ being the database size) can be evaluated within reasonable accuracy and amount of computation, not suitable for evaluating tasks like distributed learning.

Distributed learning with differential privacy is traditionally studied under the central model setting \cite{mcmahan2017learning,bonawitz2017practical, kairouz2021distributed}.
Most studies of distributed learning under the local or shuffle model have investigated only the $\eps_0$-local randomizer \cite{bhowmick2018protection,erlingsson2020encode,girgis2021differentially,girgis2021renyib}.
In \cite{liew2022shuffled}, the Gaussian mechanism is used in distributed learning, but the accounting of privacy is performed based on approximate DP, less optimal compared to our RDP approach.

\noindent\textbf{Paper organization.}
In the next section, we provide preliminaries of this paper, containing various related results of DP, as well as problem formulation.
Then, we prove our main results characterizing the RDP of shuffle Gaussian.  
In Section \ref{sec:extension}, the shuffle Gaussian mechanism is extended for applications in distributed learning.
Before closing, we give the numerical results in Section \ref{sec:numerics}.

\section{Preliminaries}\label{sec:pre}
\textbf{Problem formulation.}
Let there be $n$ users.
Let $D = (x_1, \dotsc, x_n)$ be a database of size $n$, where each of $i$, $i\in[n]$ is a data instance held by user $i$.
$x_i$ is a $d$-dimensional vector, $x_i\in \mathbb{R}^d$ and w.l.o.g., normalized to be $||x_i||_2\in [0.1]$.

Each user applies Gaussian mechanism $\mathcal{G}$ with variance $\sigma$ on $x_i$ and sends it, $\tilde{x_i}$ to the shuffler.
Let $\mathcal{S}_n:\mathcal{Y}^n\to \mathcal{Y}^n$ be the shuffling operation which takes $\tilde{x_i}$'s as input and outputs a uniformly and randomly permuted $\tilde{x_i}$'s.
The randomization mechanism $\mathcal{M}$ of interest is therefore
$$
\mathcal{M}(D) = \mathcal{S}_n (\mathcal{G}(x_1), \dotsc, \mathcal{G}(x_n)).
$$
Our purpose is to characterize and quantify the DP of $\mathcal{M}$ through RDP.

Let us next give definitions related to differential privacy and other known results. 

\begin{definition}[Central Differential Privacy \cite{Dwork14}]
A randomization mechanism, $\mathcal{M}: \mathcal{D}^n \rightarrow \calS$ with domain $\mathcal{D}^n$ and range $\mathcal{S}$ satisfies central ($\epsilon$, $\delta$)-differential privacy (DP), where $\eps \geq 0$, $\delta \in [0,1]$, if for any two adjacent databases $D, D' \in \mathcal{D}^n$ with $n$ data instances and for any subset of outputs $S \subseteq \mathcal{S}$, the following holds:
\begin{align}
\Pr[\mathcal{M}(D) \in S] \leq e^\epsilon \cdot \Pr[\mathcal{M}(D') \in S] + \delta.
\end{align}
\end{definition}
In this paper, adjacent databases are referred to as databases that have one data instance replaced by another; also known as ``replacement" DP in the literature.  
The $(\eps,\delta)$-DP mechanism is also said to satisfy \textit{approximate} DP, or is $(\eps,\delta)$-indistinguishable.
Furthermore, we simply refer to $(\eps,\delta)$-DP as ``DP" when it is unambiguous.
We also abbreviate $\mathcal{D}^n$ as $\mathcal{D}$ when the context is clear.

The randomization mechanism is known as local randomizer when the mechanism is applied to each data instance instead of the whole database.
Formally, this form of mechanism is said to satisfy local DP:
\begin{definition}[Local Differential Privacy (LDP) \cite{kasiviswanathan2011can}]
A randomization mechanism $\calA: \mathcal{D} \rightarrow \mathcal{S}$ satisfies local $(\eps,\delta)$-DP if for all pairs $x,x'\in \mathcal{D}$, $\calA(x)$ and $\calA(x')$ are $(\eps,\delta)$-indistinguishable.
\end{definition}
We often use $\eps_0$ instead of $\eps$ when referring to LDP. A mechanism satisfying $(\eps_0,0)$-LDP is also called an $\eps_0$-LDP randomizer.

Next, we define R{\'e}nyi differential privacy, the main privacy notion used throughout this work (see also \cite{dwork2016concentrated,bun2016concentrated}).
\begin{definition}[R{\'e}nyi Differential Privacy (RDP)~\cite{mironov2017renyi}]\label{def:RDP}
A randomization mechanism $\mathcal{M}: \mathcal{D} \rightarrow \mathcal{S}$ is said to satisfy  $(\lambda,\epsilon)$-RDP, for $\lambda\in(1,\infty)$, if for any adjacent databases $D$, $D'\in\calD$, the R{\'e}nyi divergence of order $\lambda$ between $\calM(D)$ and $\calM(D')$, as defined below, is upper-bounded by $\eps$:
\begin{equation}
D_{\lambda}(\calM(D)||\calM(D')):=\frac{1}{\lambda-1}\log\left(\mathbb{E}_{y\sim\calM(D')}\left[\left(\frac{\calM(D)(y)}{\calM(D')(y)}\right)^{\lambda}\right]\right)\leq \epsilon,
\label{eq:rdpdiv}
\end{equation}
where $\calM(\calD)(y)$ denotes $\calM$ taking $D$ as input to output $y$ with certain probability.
\end{definition}
We take a functional view of $\eps$, writing it as $\eps(\lambda)$.
We occasionally  call $\eps(\lambda)$ the RDP, and also refer to $D_{\lambda}(\calM(D)||\calM(D'))$ as the RDP when the context is clear.

The RDP is most useful at composing DP mechanisms, where it has cleaner composition than approximate DP. The formal description is given below.
\begin{lemma}[Adaptive RDP composition  \cite{mironov2017renyi}]
\label{lm:rdpcompose}
Let mechanisms $\calM_1,\calM_2$ taking $D\in\calD$ as input be $(\lambda,\eps_1)$,$(\lambda,\eps_2)$-RDP respectively, the composed mechanism, $\calM_1 \times \calM_2$ satisfies $(\lambda,\eps_1+\eps_2)$-RDP.
\end{lemma}
After accounting for the privacy with RDP, the RDP notion is often converted to the conventional and interpretable approximate DP notion.
The conversion is given by the following lemma.
\begin{lemma}[RDP-to-DP conversion~\cite{canonne2020discrete,balle2020hypothesis}]\label{lem:RDP_DP} 
\label{lm:rdpdp}
A mechanism $\calM$ satisfying $\left(\lambda,\epsilon\left(\lambda\right)\right)$-RDP also satisfies $\left(\epsilon,\delta\right)$-DP, where $1<\delta<0$ is arbitrary and $\eps$ is given by
\begin{equation}
\label{eq:rdpdp}
\epsilon = \min_{\lambda} \left(\epsilon\left(\lambda\right)+\frac{\log\left(1/\delta\right)+\left(\lambda-1\right)\log\left(1-1/\lambda\right)-\log\left(\lambda\right)}{\lambda-1}\right).
\end{equation}
\end{lemma}
In the above definitions, we have not specified what the underlying randomization mechanisms are.
We next give the definition of the core mechanism used in this work, the Gaussian mechanism.
\begin{definition}[Gaussian mechanism] \label{def:gauss}
Given $x\in\mathbb{R}^d$, the Gaussian mechanism applied to $x$ is a mechanism with parameter $\sigma$ that adds zero mean isotropic Gaussian perturbation of variance $\sigma^2$ to $x$, outputting $\mathcal{M}(x)= x+\mathcal{N}(0,\sigma^2I_d)$.
\end{definition}

\section{The RDP of shuffle Gaussian}
We first derive the R{\'e}nyi divergence (of order $\lambda$) of the shuffle Gaussian mechanism in this section.
Subsequently, an upper bound on the notion is given, followed by a description of the numerical techniques for evaluating the RDP.

The main result of this paper is given by the following theorem.
\begin{theorem}[Shuffle Gaussian RDP]
\label{th:main}
The R{\'e}nyi divergence (of order $\lambda$) of the shuffle Gaussian mechanism with variance $\sigma^2$ for neighboring datasets with $n$ instances is given by
$$
    D_{\lambda}(\calM(D)||\calM(D')) = \frac{1}{\lambda-1}\log\left(\frac{e^{-\lambda/2\sigma^2}}{n^\lambda}\sum_{\substack{k_1+\dotsc+k_n=\lambda;\\k_1,\dotsc,k_n\geq 0}}\binom{\lambda}{k_1,\dotsc,k_n}e^{\sum_{i=1}^nk_i^2/2\sigma^2}\right)
$$
\end{theorem}
\begin{proof}
We consider w.l.o.g. adjacent databases $D,D' \in \mathbb{R}^n$ with one-dimensional data instances $D = (0, \dotsc, 0)$, $D'= (1,0, \dotsc, 0)$. \footnote{This can be done w.l.o.g. by normalizing and performing unitary transform on the data instances. See also Appendix B.1 of \cite{koskela2020computing}.}
Note that the databases can also be represented in the form of $n$-dimensional vector: where $D$ is simply $0$, an $n$-dimensional vector with all elements equal to zero, and $D' = e_1$, where $e_i$ is a unit vector with the elements in all dimensions except the $n$-th ($i\in[n]$) one equal to zero.

The shuffle Gaussian mechanism first applies Gaussian noise (with variance $\sigma^2$) to each data instance, and subsequently shuffle the noisy instances.
The shuffled output of $D$ distributes as $\mathcal{M}(D) \sim (\mathcal{N}(0,\sigma^2), \dotsc, \mathcal{N}(0,\sigma^2))$, an $n$-tuple of $\mathcal{N}(0,\sigma^2)$.
Intuitively, the adversary sees a anonymized set of randomized data of size $n$. 
Due to the homogeneity of $D$, we can write
$\mathcal{M}(D) \sim \mathcal{N}(0,\sigma^2I_n)$ using the vector notation mentioned above. \footnote{Note that the vector space defined here is over the dataset dimension instead of input dimension as in Definition \ref{def:gauss}.}

On the other hand, since all data instances except one is 0 in $D'$, the mechanism can be written as a mixture distribution: $\mathcal{M}(D')\sim \frac{1}{n} (\mathcal{N}(e_1,\sigma^2I_n)+ \dotsc + \mathcal{N}(e_n,\sigma^2I_n))$.
Intuitively, each output has probability $1/n$ carrying the non-zero element, and hence the mechanism as a whole is a uniform mixture of $n$ distributions.
In summary, the neighbouring databases of interest are
\begin{align}
\mathcal{M}(D) &\sim \mathcal{N}(0,\sigma^2I_n),\\
\mathcal{M}(D')&\sim \frac{1}{n} \left(\mathcal{N}(e_1,\sigma^2I_n)+ \dotsc + \mathcal{N}(e_n,\sigma^2I_n)\right).
\end{align}

Note that $\mathbb{P}(\mathcal{N}(0,\sigma^2))= \exp[-x^2/2\sigma^2]/\sqrt{2\pi \sigma^2}$, and
\begin{align}
  n \cdot \left(2\pi\sigma^2\right)^{n/2}\cdot \mathbb{P}(\mathcal{M}(D')(x)) &= \exp\left[-(x_1-1)^2/2\sigma^2-\sum_{i=2}^n x_i^2/2\sigma^2\right]+\dotsc \nonumber \\
    &+ \exp\left[-(x_n-1)^2/2\sigma^2-\sum_{i=1}^{n-1} x_i^2/2\sigma^2\right].
\end{align}
That is, for $j\in[n]$, the $j$-th component of the mixture distribution contains a term proportional to $\exp{\left[-(x_j-1)^2/2\sigma^2-\sum_{i\neq j} x_i^2/2\sigma^2\right]}$.

We are concerned with calculating the following quantity, $\mathbb{E}_{x\sim \mathcal{M}(D)}\left[\left(\frac{\mathcal{M}(D')(x)}{\mathcal{M}(D)(x)}\right)^\lambda\right]$ for RDP:
\begin{align}
    &\mathbb{E}_{x\sim \mathcal{M(D)}}\left[\left(\frac{\mathcal{M}(D')(x)}{\mathcal{M}(D)(x)}\right)^\lambda\right]\nonumber \\ 
    &= \int \left(\frac{\exp{\left[-(x_1-1)^2/2\sigma^2-\sum_{i=2}^n x_i^2/2\sigma^2\right]}+\dotsc}{n\exp[-\sum_{i=1}^nx_i^2/2\sigma^2]}\right)^{
\lambda}\exp[-\sum_{i=1}^nx_i^2/2\sigma^2] \frac{d^nx}{(2\pi \sigma^2)^{n/2}} \label{eq:main1} \\
&= \int \left(\frac{\exp{\left[-(x_1-1)^2/2\sigma^2+x_1^2/2\sigma^2\right]}+\dotsc}{n}\right)^{
\lambda}\exp[-\sum_{i=1}^nx_i^2/2\sigma^2] \frac{d^nx}{(2\pi \sigma^2)^{n/2}} \label{eq:main2}\\
&= \int \left(\exp{\left[(2x_1-1)/2\sigma^2\right]} + \dotsc \right)^{\lambda}\exp[-\sum_{i=1}^nx_i^2/2\sigma^2]\frac{d^nx}{n^\lambda(2\pi \sigma^2)^{n/2}} \label{eq:main3} \\
&= \int \left(\sum_{i=1}^n\exp{\left[(2x_i-1)/2\sigma^2\right]}\right)^{\lambda}\exp[-\sum_{i=1}^nx_i^2/2\sigma^2]\frac{d^nx}{n^\lambda(2\pi \sigma^2)^{n/2}}.\label{eq:main4}
\end{align}
Let us explain the above calculation in detail.
We first notice that $\mathbb{E}_{x\sim \mathcal{M(D)}}\left[\left(\frac{\mathcal{M}(D')(x)}{\mathcal{M}(D)(x)}\right)^\lambda\right]$ is an $n$-th dimensional integral of $x_i$ ($i\in[n]$), as shown in Equation \ref{eq:main1}.
For each term $j\in[n]$ in the nominator of $(\dotsc)^\lambda$, we divide it by the denominator $\exp[-\sum_{i=1}^nx_i^2/2\sigma^2]$, yielding $\exp[-(x_j-1)^2/2\sigma^2-\sum_{i\neq j}^nx_i^2/2\sigma^2+\sum_{i=1}^nx_i^2/2\sigma^2]= \exp[-(x_j-1)^2/2\sigma^2 + x_j^2/2\sigma^2]$, with all $x_i$ terms in $i\in[n]$ except $j$ canceled out (Equation \ref{eq:main2}).
The term can be further simplified to $\exp[(2x_j-1)/2\sigma^2]$ as in Equation \ref{eq:main4}.

Then, we expand the expression $(\dotsc)^\lambda$ using the multinomial theorem:
\begin{align}
\label{eq:k-integral}
    \left(\sum_{i=1}^n\exp{\left[(2x_i-1)/2\sigma^2\right]}\right)^{\lambda} = \sum_{\substack{k_1+...+k_n=\lambda;\\k_1,\dotsc,k_n\geq 0}}\binom{\lambda}{k_1,\dotsc,k_n}\prod_{i=1}^n\exp\left[k_i(2x_i-1)/2\sigma^2\right],
\end{align}
where $\binom{\lambda}{k_1,\dotsc,k_n} = \frac{\lambda!}{k_1!k_2!\dotsc k_n!}$ is the multinomial coefficient, and $k_i \in \mathbb{Z^+}$ for $i\in[n]$.

Before proceeding, we make a detour to calculate the following integral (with $k\in \mathbb{Z^+}$):
\begin{align}
    &\int \exp[k(2x-1)/2\sigma^2]\exp[-x^2/2\sigma^2]\frac{dx}{\sqrt{2\pi\sigma^2}} \nonumber \\
    &= \int \exp[-(x-k)^2/2\sigma^2 +(k^2-k)/2\sigma^2]\frac{dx}{\sqrt{2\pi\sigma^2}} \nonumber \\
    & = \exp[(k^2-k)/2\sigma^2]. \nonumber
\end{align}
Using the above expression and Equation \ref{eq:k-integral}, we can write Equation \ref{eq:main4} as
\begin{align}
&\int \left(\sum_{i=1}^n\exp{\left[(2x_i-1)/2\sigma^2\right]}\right)^{\lambda}\exp[-\sum_{i=1}^nx_i^2/2\sigma^2]\frac{d^nx}{n^\lambda(2\pi \sigma^2)^{n/2}} \nonumber\\
    &=\int\sum_{\substack{k_1+...+k_n=\lambda;\\k_1,\dotsc,k_n\geq 0}} \binom{\lambda}{k_1,\dotsc,k_n}\prod_{i=1}^n\exp\left[k_i(2x_i-1)/2\sigma^2\right] \frac{\exp[-x_i^2/2\sigma^2]d^nx}{n^\lambda(2\pi \sigma^2)^{n/2}} \nonumber \\
    &= \frac{1}{n^\lambda}\sum_{\substack{k_1+...+k_n=\lambda;\\k_1,\dotsc,k_n\geq 0}} \binom{\lambda}{k_1,\dotsc,k_n}\prod_{i=1}^n\exp[(k_i^2-k_i)/2\sigma^2],
\end{align}
where we have moved the $1/n^\lambda$ factor to the front.
Noticing that $\prod_{i=1}^n\exp[-k_i] = \exp[-\sum_i^{n}k_i]$ and that $\sum_{i=1}^nk_i=\lambda$ under the multinomial constraint, we have
\begin{align}
\label{eq:main5}
& \frac{1}{n^\lambda}\sum_{\substack{k_1+...+k_n=\lambda;\\k_1,\dotsc,k_n\geq 0}} \binom{\lambda}{k_1,\dotsc,k_n}\prod_{i=1}^n\exp[(k_i^2-k_i)/2\sigma^2]\\
    &=\frac{e^{-\lambda/2\sigma^2}}{n^\lambda}\sum_{\substack{k_1+\dotsc+k_n=\lambda;\\k_1,\dotsc,k_n\geq 0}}\binom{\lambda}{k_1,\dotsc,k_n}e^{\sum_{i=1}^nk_i^2/2\sigma^2}.
\end{align}
Combining the above expression with Equation \ref{eq:rdpdiv}, we obtain the expression given in Equation \ref{eq:main} as desired.
\end{proof}

Using the above results, we next give an upper bound on the shuffle Gaussian RDP.
\begin{corollary}[Upper bound of shuffle Gaussian RDP]
\label{co:upper}
The shuffle Gaussian RDP $\eps(\lambda)$ is upper-bounded by $\lambda/2\sigma^2$.
\end{corollary}
\begin{proof}
From the Cauchy-Schwartz inequality, $\sum_{i=1}^n k_i^2 \leq (\sum_{i=1}^n k_i)^2 = \lambda^2$.
Then, from Equation \ref{eq:main5},
\begin{align*}
    \frac{e^{-\lambda/2\sigma^2}}{n^\lambda}&\sum_{\substack{k_1+\dotsc+k_n=\lambda;\\k_1,\dotsc,k_n\geq 0}}\binom{\lambda}{k_1,\dotsc,k_n}e^{\sum_{i=1}^nk_i^2/2\sigma^2} \\ &\leq \frac{e^{-\lambda/2\sigma^2}}{n^\lambda}\sum_{\substack{k_1+\dotsc+k_n=\lambda;\\k_1,\dotsc,k_n\geq 0}}\binom{\lambda}{k_1,\dotsc,k_n}e^{\lambda^2/2\sigma^2}\\
    &= e^{(\lambda^2-\lambda)/2\sigma^2},
\end{align*}
using in the last line,  $\sum_{\substack{k_1+\dotsc+k_n=\lambda;\\k_1,\dotsc,k_n\geq 0}}\binom{\lambda}{k_1,\dotsc,k_n}=n^\lambda$.
Substituting this to Equation \ref{eq:rdpdiv}, we get the desired expression.
\end{proof}
This corollary states that the upper bound of the shuffle Gaussian RDP is equal to the Gaussian RDP without shuffling \cite{mironov2017renyi}.
Let us emphasize that this is a strictly tight upper bound compared with previous studies.
In early analyses of shuffling \cite{erlingsson2019amplification,balle2020privacy}, the resulting $\eps$ after shuffling can be larger than the local one, $\eps_0$, especially at the large $\eps_0$ region.
In \cite{feldman2022hiding}, the shuffle $\eps$ is shown, only numerically, to saturate at $\eps_0$.
We have instead proved in the above corollary, at least for the Gaussian mechanism, that the upper bound is analytically and strictly the same as the one without shuffling.
This also means that there is no such phenomenon as ``privacy degradation"  due to shuffling in Gaussian mechanism.

\subsection{Evaluating the shuffle Gaussian RDP numerically}
Evaluating the shuffle Gaussian RDP numerically is non-trivial, as $n$ can be as large as $10^5$ in distributed learning scenarios and $\lambda$ can be large as well.
Here, we describe our numerical recipes for calculating the RDP.

We exploit the permutation invariance inherent in Equation \ref{eq:main} to perform more efficient computation.
Equation \ref{eq:main} contains a sum of multinomial coefficients multiplied by an expression with integer $k_i$ ($i\in[n]$) constrained by $k_1 + \dotsc +k_n = \lambda$.
We first obtain all possible $k_i$'s with the above constraint, without counting same summands differing only in the order.
This is integer partition, a subset sum problem known to be NP \cite{kleinberg2006algorithm}.
It is also related to partition in number theory \cite{andrews2004integer}.
Nevertheless, the calculation can be performed efficiently (e.g., \cite{kelleher2009generating}) and one can cache the partitions to save computational time.


Note that there is a one-to-one mapping between the term $e^{\sum_{i=1}^nk_i^2/2\sigma^2}$ in Equation \ref{eq:main5} and the partition obtained above.
We can therefore calculate the number of permutation corresponding to each partition, multiplied by the multinomial coefficient, and $e^{\sum_{i=1}^nk_i^2/2\sigma^2}$, and sum over all the partitions to obtain the numerical expression of Equation $\ref{eq:main5}$.
For a partition with the number of repetition of exponents with the same degrees (or the counts of unique $k_1,\dotsc,k_n$) being $\kappa_1, \dotsc, \kappa_{l}$ ($l\leq n$), the number of permutation can be calculated to be $n!/(\kappa_1!\dotsc \kappa_l!)$.
See Appendix \ref{app:algo} for details.
The algorithm of calculating the RDP can also be found in Algorithm \ref{alg:num}.

Since Equation \ref{eq:main} is of the form of a logarithm, evaluation with the log-sum-exp technique leads to more stable results.
Moreover, the factorials appearing in the multinomial coefficients can be calculated efficiently with the log gamma function (implemented in \texttt{scipy}).

\begin{algorithm}[t]
\caption{Numerical evaluation of Equation \ref{eq:main}.}
\label{alg:num}
\begin{algorithmic}[1]
\State \textbf{Inputs:} Database size $n$, RDP order $\lambda$, variance of Gaussian mechanism $\sigma^2$.
\State \textbf{Output:} RDP parameter $\eps$.
\State \textbf{Initialize:} $\Gamma:\{\emptyset\}$, $\texttt{Sum}=0$.
\State $\Gamma \leftarrow \texttt{GetPartition}(\lambda)$
\For { $\rho \in \Gamma$}
\State $\{\kappa_i\} \leftarrow \texttt{GetUniqueCount}(\rho)$
\State $\texttt{Sum} \leftarrow \texttt{Sum} + \frac{e^{-\lambda/2\sigma^2}}{n^\lambda}\binom{\lambda}{k_1\dotsc k_n}\frac{n!}{\kappa_1!\dotsc\kappa_i!} e^{\sum_{j=1}^{n}k_j^2}$ where $\{k_i\}\in \rho$
\EndFor
\State \textbf{Return:} $\eps = \texttt{Sum}$
\end{algorithmic}
\end{algorithm}

\section{Extensions of the Shuffle Gaussian}
\label{sec:extension}
In this section, we extend the study of the shuffle Gaussian RDP to mechanisms more attuned to distributed learning.
Specifically, we consider mechanisms where only a subset of the dataset participates in training, typical for large-scale distributed learning.
Throughout our study, we consider protecting the privacy of users each holding a single data instance.

\subsection{Subsampled shuffle Gaussian mechanism}
\label{subsec:subshuff}
We first consider the mechanism where a fixed number of users, $m$, is sampled randomly from $n$ available users to participate in training.
The sampled users randomize their data with Gaussian noises and pass them to a shuffler.
We call this overall procedure the subsampled (without replacement) shuffle Gaussian mechanism. \footnote{See also  \cite{girgis2021shuffled}, where a similar protocol except that an $\eps_0$-local randomizer is utilized, has been considered.}.

The RDP of the subsampled shuffle Gaussian mechanism is given by the following theorem.
\begin{theorem}[Subsampled Shuffle Gaussian RDP]
\label{th:subshuff}
Let $n$ be the total number of users, $m$ be the number of subsampled users in each round of training, and $\gamma$ be the subsampling rate, $\gamma = m/n$. Also let $\eps^{\rm SG}_m(\lambda)$ be the shuffle Gaussian RDP given in Theorem \ref{th:main}.
The subsampled shuffle Gaussian mechanism satisfies $(\lambda,\eps^{\rm SSG}_{\gamma,m}(\lambda))$-RDP, where 
\begin{align}
\label{eq:ssg}
    \eps^{\rm SSG}_{\gamma,m}(\lambda) \leq \frac{1}{\lambda -1}
\log\left(1 + \gamma^2 \binom{\lambda}{2} \min\left\{4(e^{\eps^{\rm SG}_m(2)}-1), 2 e^{\eps^{\rm SG}_m(2)}\right\}
+\sum_{j=3}^{\lambda}2\gamma^j \binom{\lambda}{j} e^{(j-1)\eps^{\rm SG}_m(j)}
\right)
\end{align}
\end{theorem}
\begin{proof}[Proof sketch]
The proof is a direct application of Theorem \ref{eq:main} and the subsampled RDP results of \cite{wang2019subsampled}.
The full proof can be found in Appendix.
\end{proof}

\subsection{Shuffled Check-in Gaussian mechanism}
\label{subsec:sci}
In practice, it is difficult to achieve uniform subsampling with a fixed number of users via the orchestrating server \cite{balle2020privacy,girgis2021differentially}.
It is more natural to let the user decide (with her own randomness) whether to participate in training.
This is known as the shuffled check-in mechanism \cite{liew2022shuffled}.
The mechanism may also be considered as the shuffle version of Poisson subsampling in the literature.

Let $n$ be the total number of users, and $\gamma$ be the check-in rate (the probability of a user participating in training)
The shuffled check-in's RDP is given as follows \cite{liew2022shuffled}.
\begin{equation}
\label{eq:sci}
    \eps^{\rm SCI}(\lambda) \leq \frac{1}{\lambda-1} \log \left(
    \sum_{k=1}^{n}  \binom{n}{k} \gamma^{k} (1-\gamma)^{n-k} e^{(\lambda-1)\eps^{\rm SSG}_{\gamma,k}(\lambda)}
    \right).
\end{equation}
Here, $\eps^{\rm SSG}_{\gamma,k}(\lambda)$ is the subsampled shuffle Gaussian RDP defined in Equation \ref{eq:ssg}, with $k$ the number of shuffled instances, $\gamma$ the subsampling rate. 

The direct evaluation of the above equation is inefficient numerically, as it involves a sum over $n$.
We opt for a more efficient computation, which requires the use of the following conjecture.

\begin{conjecture}[Monotonicity]
\label{cj:mono}
The expression in Equation \ref{eq:main5}, that is,
$$
\frac{e^{-\lambda/2\sigma^2}}{n^\lambda}\sum_{\substack{k_1+\dotsc+k_n=\lambda;\\k_1,\dotsc,k_n\geq 0}}\binom{\lambda}{k_1,\dotsc,k_n}e^{\sum_{i=1}^nk_i^2/2\sigma^2},
$$
is a monotonically decreasing function with respect to $n$.
\end{conjecture}

Although we are not able to prove rigorously the above-mentioned conjecture yet, our empirical scan over $n$ shows that it is indeed monotonically decreasing (at least within our scope of investigation).
With this conjecture, we present the following theorem for calculating the shuffled check-in Gaussian RDP with efficiency.

\begin{theorem}[Shuffled Check-in Gaussian RDP]
\label{th:sci}
Let $n$ be the total number of users, and $\gamma$ be the check-in rate. Also let $\eps^{\rm SSG}_{\gamma,m}(\lambda)$ be the subsampled shuffle Gaussian RDP defined in Equation \ref{eq:ssg}, with $k$ the number of shuffled instances, $\gamma$ the subsampling rate.
If Conjecture \ref{cj:mono} holds, the shuffled check-in Gaussian mechanism satisfies $(\lambda,\eps^{\rm SCI}(\lambda))$-RDP, where
\begin{align}
\label{eq:scifast}
\eps^{\rm SCI}(\lambda) \leq \frac{1}{\lambda-1}\log\left( e^{(\lambda-1)\eps^{\rm SSG}_{\gamma,1}(\lambda)-\Delta^2n\gamma/2} + e^{(\lambda-1)\eps^{\rm SSG}_{\gamma,(1-\Delta)n\gamma+1}(\lambda)}\right).
\end{align}
Here, $\Delta \in [0,1]$ is arbitrary.
\end{theorem}
\begin{proof}[Proof sketch]
The proof follows closely to the techniques used in \cite{liew2022shuffled}, where monotonicity and the Chernoff bound are utilized to prove the theorem. The full proof can be found in Appendix \ref{app:proof}.
\end{proof}
The main advantage of using Equation \ref{eq:scifast} is that the calculation of the RDP with this theorem is $O(n)$ times more efficient computationally compared to Equation \ref{eq:sci}.\section{Numerical Results}\label{sec:numerics}
In this section, we present the numerical evaluation results of the shuffle Gaussian and its extensions.
We also use the latter mechanisms to perform machine learning tasks under the distributed setting.

\subsection{Numerical results of the shuffle Gaussian}
An upper bound on the shuffle model with an $(\eps,\delta)$-LDP randomizer is given by the clones method \cite{feldman2022hiding}, which is our target of comparison.
According to \cite{feldman2022hiding}, given $n$ instances undergone $(\eps_0,\delta_0)$-LDP randomization and further shuffled such that $\eps_0 \leq \log(\frac{n}{16\log(2\delta})$ for any $\delta \in [0,1]$, the overall DP is $(\epsilon, \delta + (s^{\eps}+1))(1+e^{-\eps_0}/2)n\delta_0)$-DP.
Here, $\eps = O(\eps_0 \frac{\sqrt{e^{\eps_0}\log(1/\delta)}}{\sqrt{n}})$ when $\eps_0 > 1$.

The Gaussian mechanism is however governed by the noise variance, $\sigma$, and there is an infinite pair of corresponding $(\eps, \delta)$.
To use the clones method, we fix $\sigma$, $\delta_0=\delta$ and scan over $\delta_0$ to find parameters satisfying $\delta_{\rm final} \leq 1/n$, which is the overall DP $\delta$ parameter.
Finally, mechanism composition is performed with the strong composition theorem \cite{kairouz2015composition}.

In Table \ref{tab:clonevsours}, we compare privacy accounting using the clones method with ours (with RDP), fixing $n=60,000$, $\sigma = 9.48$ (corresponding to $(1,1/n)$-LDP) and global sensitivity 1.
We set the maximum $\lambda$ to be evaluated as 30.
While the clones method gives tighter result without composition, the RDP accounting shows tighter results by a large margin with composition.

\begin{table}
\caption{the value of shuffle $\eps$ under composition. The clones method \cite{feldman2022hiding} is compared with our RDP-based method.}
\label{tab:clonevsours}
\begin{center}
\small
\begin{tabular}{lrrrrrrr}
\hline
 No. of composition   &       1 &       2 &       3 &       4 &       5 &       6 &       7       \\
\hline
 Clones  & 0.18623 & 0.38461 & 0.59516 & 0.79355 & 1.02241 & 1.22689 & 1.43138  \\
 Ours    & 0.22820  & 0.22820  & 0.22821 & 0.22821 & 0.22821 & 0.22822 & 0.22822  \\
\hline
\end{tabular}
\end{center}
\end{table}

Let us plot the $\eps$ dependence at larger number of composition.
In Figure \ref{fig:shuff}, we show the RDP and upper bound, again demonstrating significant amplification due to shuffling.

\subsection{Distributed learning}
We perform an experiment under the shuffled check-in protocol with a setup similar to the one given in \cite{liew2022shuffled}.
We use the MNIST handwritten digit dataset \cite{lecun2010mnist}, with each user holding one data instance from it.
A convolutional neural network is used as the target model $\theta \in \mathbb{R}^d$ under the distributed learning setting.
In each round, each participating user calculates the gradient $g\in\mathbb{R}^d$ using the received model.
The user clips the gradient, $g\leftarrow g \cdot{\rm min}\left\{1, \frac{C}{|g|_2}\right\}$ with a clipping factor $C$, and add Gaussian noise to it, $\tilde{g}\to g + \mathcal{N}(0,C^2\sigma^2I_d)$ before sending it to the server.

The total number of user is 60,000, i.e., the size of the MNIST train dataset.
The check-in rate is set to be $\gamma=0.1$, $C$ 0.01, $\sigma = 5$.
The test dataset is used to evaluate the classification accuracy.

We compare our approach with two methods introduced in \cite{liew2022shuffled}.
The first method performs conservatively the privacy accounting of approximate DP of subsampling and shuffling (Baseline).
The second method converts approximate DP of (subsampled) shuffling to RDP, and perform the privacy accounting under the shuffled check-in framework (generic bound).

We note that the experiments of interest can be simulated conveniently by making simple modifications to existing libraries that implement DP-SGD (e.g., \texttt{Opacus} \cite{yousefpour2021opacus}).
This is done by noticing that in DP-SGD, given a batch of clipped gradients $\tilde{g}$, Gaussian noise is applied to the sum of the gradients, i.e., 
$$\sum_{i=1}^{B}\tilde{g}_i \leftarrow \sum_{i=1}^{B}\tilde{g}_i + \mathcal{N}(0,\sigma^2)$$
Here, $B$ is the number of samples.
In shuffle Gaussian or local Gaussian, the server sums over the received gradients randomized with Gaussian noises,
$$\sum_{i=1}^{B}\left\{\tilde{g}_i + \mathcal{N}(0, \sigma^2) \right\} = \sum_{i=1}^{B}\tilde{g}_i + \mathcal{N}(0,B\sigma^2)$$
Therefore, we can simply multiply $\sigma$ by a factor of $\sqrt{B}$ in existing DP-SGD libraries to simulate the decentralized setting of shuffle/local Gaussian.

Figure \ref{fig:feddpsgd} shows the accuracy-versus-$\eps$ plot of our experiment.
Here, a total number of 5,540 rounds of training has been run.
Our accounting requires less $\eps$ to reach optimal accuracy ($80\%$).
Note also that the generic accounting method in \cite{liew2022shuffled} is rather sophisticated (further requiring an inner round of optimization) and requires $O(n)$ times more computation compared to ours.

\begin{figure}
\begin{subfigure}{0.5\textwidth}
        \centering 
    \includegraphics[width=0.8\textwidth]{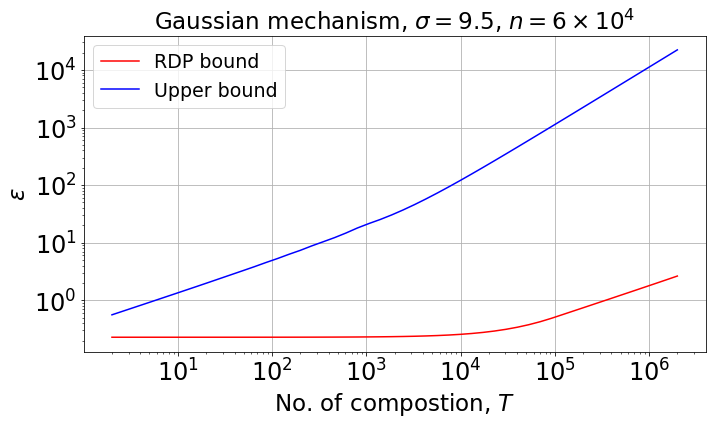}
      \caption{Shuffling effects}
      \label{fig:shuff}
    \end{subfigure}\hfil 
    \begin{subfigure}{0.5\textwidth}
        \centering 
    \includegraphics[width=0.8\textwidth]{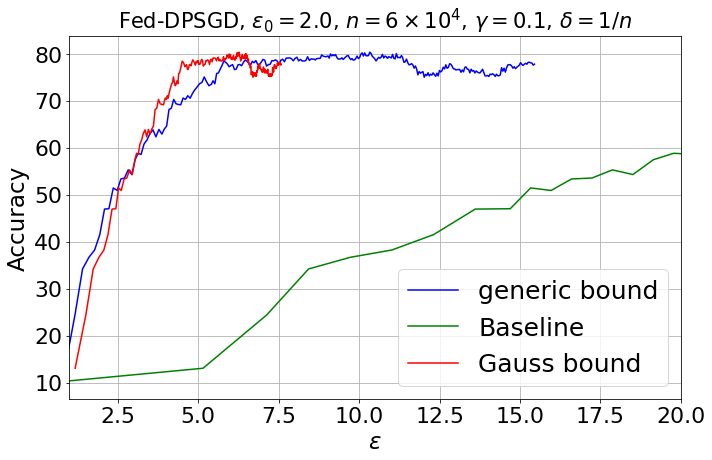}
      \caption{Distributed learning}
      \label{fig:feddpsgd}
    \end{subfigure} 
  \caption{The left figure shows the upper bound (no shuffling) and RDP bound (in red) derived for the shuffle Gaussian mechanism varying the number of composition.
  The right figure shows the accuracy versus $\epsilon$ under the distributed learning setting.
  ``Gauss bound" (in red) is the bound derived in this work.}
      \label{fig:result}
\end{figure}

\section{Conclusion}
In this paper, we have analyzed the RDP of shuffle Gaussian mechanism, and applied it to distributed learning.
There are still some open questions.
Calculating the RDP at large $\lambda$ is computationally infeasible.
It has deep relation with number theory and looking for hints from there to calculate these values is a possible future direction.
A rigorous proof of Conjecture \ref{cj:mono} is also needed.

As various settings of the Gaussian mechanism under the centralized learning setting, i.e., DP-SGD, are relatively well studied,
we hope that our privacy accountant of shuffle Gaussian opens an avenue that  inspires more in-depth studies on various datasets and settings of distributed learning under the shuffle model, using the relatively simple Gaussian mechanism, similar to DP-SGD.

\bibliographystyle{plain}
\bibliography{ref}

\begin{thebibliography}{10}

\bibitem{abadi2016deep}
Martin Abadi, Andy Chu, Ian Goodfellow, H~Brendan McMahan, Ilya Mironov, Kunal
  Talwar, and Li~Zhang.
\newblock Deep learning with differential privacy.
\newblock In {\em Proceedings of the 2016 ACM SIGSAC Conference on Computer and
  Communications Security}, pages 308--318. ACM, 2016.

\bibitem{andrews2004integer}
George~E Andrews and Kimmo Eriksson.
\newblock {\em Integer partitions}.
\newblock Cambridge University Press, 2004.

\bibitem{balle2020hypothesis}
Borja Balle, Gilles Barthe, Marco Gaboardi, Justin Hsu, and Tetsuya Sato.
\newblock Hypothesis testing interpretations and renyi differential privacy.
\newblock In {\em International Conference on Artificial Intelligence and
  Statistics}, pages 2496--2506. PMLR, 2020.

\bibitem{balle2019privacy}
Borja Balle, James Bell, Adri{\`a} Gasc{\'o}n, and Kobbi Nissim.
\newblock The privacy blanket of the shuffle model.
\newblock In {\em Annual International Cryptology Conference}, pages 638--667.
  Springer, 2019.

\bibitem{balle2020privacy}
Borja Balle, Peter Kairouz, Brendan McMahan, Om~Thakkar, and Abhradeep
  Guha~Thakurta.
\newblock Privacy amplification via random check-ins.
\newblock {\em Advances in Neural Information Processing Systems},
  33:4623--4634, 2020.

\bibitem{bassily2014private}
Raef Bassily, Adam Smith, and Abhradeep Thakurta.
\newblock Private empirical risk minimization: Efficient algorithms and tight
  error bounds.
\newblock In {\em 2014 IEEE 55th Annual Symposium on Foundations of Computer
  Science}, pages 464--473. IEEE, 2014.

\bibitem{bhowmick2018protection}
Abhishek Bhowmick, John Duchi, Julien Freudiger, Gaurav Kapoor, and Ryan
  Rogers.
\newblock Protection against reconstruction and its applications in private
  federated learning.
\newblock {\em arXiv preprint arXiv:1812.00984}, 2018.

\bibitem{bittau2017prochlo}
Andrea Bittau, {\'U}lfar Erlingsson, Petros Maniatis, Ilya Mironov, Ananth
  Raghunathan, David Lie, Mitch Rudominer, Ushasree Kode, Julien Tinnes, and
  Bernhard Seefeld.
\newblock Prochlo: Strong privacy for analytics in the crowd.
\newblock In {\em Proceedings of the 26th Symposium on Operating Systems
  Principles}, pages 441--459, 2017.

\bibitem{bonawitz2017practical}
Keith Bonawitz, Vladimir Ivanov, Ben Kreuter, Antonio Marcedone, H~Brendan
  McMahan, Sarvar Patel, Daniel Ramage, Aaron Segal, and Karn Seth.
\newblock Practical secure aggregation for privacy-preserving machine learning.
\newblock In {\em proceedings of the 2017 ACM SIGSAC Conference on Computer and
  Communications Security}, pages 1175--1191, 2017.

\bibitem{bun2016concentrated}
Mark Bun and Thomas Steinke.
\newblock Concentrated differential privacy: Simplifications, extensions, and
  lower bounds.
\newblock In {\em Theory of Cryptography Conference}, pages 635--658. Springer,
  2016.

\bibitem{canonne2020discrete}
Cl{\'e}ment~L Canonne, Gautam Kamath, and Thomas Steinke.
\newblock The discrete gaussian for differential privacy.
\newblock {\em Advances in Neural Information Processing Systems},
  33:15676--15688, 2020.

\bibitem{chaum1981untraceable}
David~L Chaum.
\newblock Untraceable electronic mail, return addresses, and digital
  pseudonyms.
\newblock {\em Communications of the ACM}, 24(2):84--90, 1981.

\bibitem{cheu2021differential}
Albert Cheu.
\newblock Differential privacy in the shuffle model: A survey of separations.
\newblock {\em arXiv preprint arXiv:2107.11839}, 2021.

\bibitem{cheu2019distributed}
Albert Cheu, Adam Smith, Jonathan Ullman, David Zeber, and Maxim Zhilyaev.
\newblock Distributed differential privacy via shuffling.
\newblock In {\em Annual International Conference on the Theory and
  Applications of Cryptographic Techniques}, pages 375--403. Springer, 2019.

\bibitem{duchi2018minimax}
John~C Duchi, Michael~I Jordan, and Martin~J Wainwright.
\newblock Minimax optimal procedures for locally private estimation.
\newblock {\em Journal of the American Statistical Association},
  113(521):182--201, 2018.

\bibitem{dwork2006our}
Cynthia Dwork, Krishnaram Kenthapadi, Frank McSherry, Ilya Mironov, and Moni
  Naor.
\newblock Our data, ourselves: Privacy via distributed noise generation.
\newblock In {\em Eurocrypt}, volume 4004, pages 486--503. Springer, 2006.

\bibitem{dwork2006calibrating}
Cynthia Dwork, Frank McSherry, Kobbi Nissim, and Adam Smith.
\newblock Calibrating noise to sensitivity in private data analysis.
\newblock In {\em Theory of cryptography conference}, pages 265--284. Springer,
  2006.

\bibitem{Dwork14}
Cynthia Dwork, Aaron Roth, et~al.
\newblock The algorithmic foundations of differential privacy.
\newblock {\em Found. Trends Theor. Comput. Sci.}, 9(3-4):211--407, 2014.

\bibitem{dwork2016concentrated}
Cynthia Dwork and Guy~N Rothblum.
\newblock Concentrated differential privacy.
\newblock {\em arXiv preprint arXiv:1603.01887}, 2016.

\bibitem{dwork2010boosting}
Cynthia Dwork, Guy~N Rothblum, and Salil Vadhan.
\newblock Boosting and differential privacy.
\newblock In {\em 2010 IEEE 51st Annual Symposium on Foundations of Computer
  Science}, pages 51--60. IEEE, 2010.

\bibitem{erlingsson2020encode}
{\'U}lfar Erlingsson, Vitaly Feldman, Ilya Mironov, Ananth Raghunathan, Shuang
  Song, Kunal Talwar, and Abhradeep Thakurta.
\newblock Encode, shuffle, analyze privacy revisited: Formalizations and
  empirical evaluation.
\newblock {\em arXiv preprint arXiv:2001.03618}, 2020.

\bibitem{erlingsson2019amplification}
{\'U}lfar Erlingsson, Vitaly Feldman, Ilya Mironov, Ananth Raghunathan, Kunal
  Talwar, and Abhradeep Thakurta.
\newblock Amplification by shuffling: From local to central differential
  privacy via anonymity.
\newblock In {\em Proceedings of the Thirtieth Annual ACM-SIAM Symposium on
  Discrete Algorithms}, pages 2468--2479. SIAM, 2019.

\bibitem{feldman2022hiding}
Vitaly Feldman, Audra McMillan, and Kunal Talwar.
\newblock Hiding among the clones: A simple and nearly optimal analysis of
  privacy amplification by shuffling.
\newblock In {\em 2021 IEEE 62nd Annual Symposium on Foundations of Computer
  Science (FOCS)}, pages 954--964. IEEE, 2022.

\bibitem{girgis2021renyib}
Antonious Girgis, Deepesh Data, and Suhas Diggavi.
\newblock Renyi differential privacy of the subsampled shuffle model in
  distributed learning.
\newblock {\em Advances in Neural Information Processing Systems}, 34, 2021.

\bibitem{girgis2021shuffled}
Antonious Girgis, Deepesh Data, Suhas Diggavi, Peter Kairouz, and
  Ananda~Theertha Suresh.
\newblock Shuffled model of differential privacy in federated learning.
\newblock In {\em International Conference on Artificial Intelligence and
  Statistics}, pages 2521--2529. PMLR, 2021.

\bibitem{girgis2021differentially}
Antonious~M Girgis, Deepesh Data, and Suhas Diggavi.
\newblock Differentially private federated learning with shuffling and client
  self-sampling.
\newblock In {\em 2021 IEEE International Symposium on Information Theory
  (ISIT)}, pages 338--343. IEEE, 2021.

\bibitem{girgis2021renyia}
Antonious~M Girgis, Deepesh Data, Suhas Diggavi, Ananda~Theertha Suresh, and
  Peter Kairouz.
\newblock On the renyi differential privacy of the shuffle model.
\newblock In {\em Proceedings of the 2021 ACM SIGSAC Conference on Computer and
  Communications Security}, pages 2321--2341, 2021.

\bibitem{gopi2021numerical}
Sivakanth Gopi, Yin~Tat Lee, and Lukas Wutschitz.
\newblock Numerical composition of differential privacy.
\newblock {\em Advances in Neural Information Processing Systems},
  34:11631--11642, 2021.

\bibitem{imola2022differentially}
Jacob Imola, Takao Murakami, and Kamalika Chaudhuri.
\newblock Differentially private subgraph counting in the shuffle model.
\newblock {\em arXiv preprint arXiv:2205.01429}, 2022.

\bibitem{kairouz2021distributed}
Peter Kairouz, Ziyu Liu, and Thomas Steinke.
\newblock The distributed discrete gaussian mechanism for federated learning
  with secure aggregation.
\newblock In {\em International Conference on Machine Learning}, pages
  5201--5212. PMLR, 2021.

\bibitem{kairouz2021advances}
Peter Kairouz, H~Brendan McMahan, Brendan Avent, Aur{\'e}lien Bellet, Mehdi
  Bennis, Arjun~Nitin Bhagoji, Kallista Bonawitz, Zachary Charles, Graham
  Cormode, Rachel Cummings, et~al.
\newblock Advances and open problems in federated learning.
\newblock {\em Foundations and Trends{\textregistered} in Machine Learning},
  14(1--2):1--210, 2021.

\bibitem{kairouz2015composition}
Peter Kairouz, Sewoong Oh, and Pramod Viswanath.
\newblock The composition theorem for differential privacy.
\newblock In {\em International conference on machine learning}, pages
  1376--1385. PMLR, 2015.

\bibitem{kasiviswanathan2011can}
Shiva~Prasad Kasiviswanathan, Homin~K Lee, Kobbi Nissim, Sofya Raskhodnikova,
  and Adam Smith.
\newblock What can we learn privately?
\newblock {\em SIAM Journal on Computing}, 40(3):793--826, 2011.

\bibitem{kelleher2009generating}
Jerome Kelleher and Barry O'Sullivan.
\newblock Generating all partitions: a comparison of two encodings.
\newblock {\em arXiv preprint arXiv:0909.2331}, 2009.

\bibitem{kleinberg2006algorithm}
Jon Kleinberg and Eva Tardos.
\newblock {\em Algorithm design}.
\newblock Pearson Education India, 2006.

\bibitem{koskela2021tight}
Antti Koskela, Mikko~A Heikkil{\"a}, and Antti Honkela.
\newblock Tight accounting in the shuffle model of differential privacy.
\newblock {\em arXiv preprint arXiv:2106.00477}, 2021.

\bibitem{koskela2020computing}
Antti Koskela, Joonas J{\"a}lk{\"o}, and Antti Honkela.
\newblock Computing tight differential privacy guarantees using fft.
\newblock In {\em International Conference on Artificial Intelligence and
  Statistics}, pages 2560--2569. PMLR, 2020.

\bibitem{lecun2010mnist}
Yann LeCun, Corinna Cortes, and CJ~Burges.
\newblock Mnist handwritten digit database.
\newblock {\em ATT Labs [Online]. Available: http://yann.lecun.com/exdb/mnist},
  2, 2010.

\bibitem{liew2022shuffled}
Seng~Pei Liew, Satoshi Hasegawa, and Tsubasa Takahashi.
\newblock Shuffled check-in: Privacy amplification towards practical
  distributed learning.
\newblock {\em arXiv preprint arXiv:2206.03151}, 2022.

\bibitem{liew2022network}
Seng~Pei Liew, Tsubasa Takahashi, Shun Takagi, Fumiyuki Kato, Yang Cao, and
  Masatoshi Yoshikawa.
\newblock Network shuffling: Privacy amplification via random walks.
\newblock {\em arXiv preprint arXiv:2204.03919}, 2022.

\bibitem{mcmahan2017communication}
Brendan McMahan, Eider Moore, Daniel Ramage, Seth Hampson, and Blaise~Aguera
  y~Arcas.
\newblock Communication-efficient learning of deep networks from decentralized
  data.
\newblock In {\em Artificial intelligence and statistics}, pages 1273--1282.
  PMLR, 2017.

\bibitem{mcmahan2017learning}
H~Brendan McMahan, Daniel Ramage, Kunal Talwar, and Li~Zhang.
\newblock Learning differentially private recurrent language models.
\newblock {\em arXiv preprint arXiv:1710.06963}, 2017.

\bibitem{mironov2017renyi}
Ilya Mironov.
\newblock R{\'e}nyi differential privacy.
\newblock In {\em 2017 IEEE 30th Computer Security Foundations Symposium
  (CSF)}, pages 263--275. IEEE, 2017.

\bibitem{murtagh2016complexity}
Jack Murtagh and Salil Vadhan.
\newblock The complexity of computing the optimal composition of differential
  privacy.
\newblock In {\em Theory of Cryptography Conference}, pages 157--175. Springer,
  2016.

\bibitem{song2013stochastic}
Shuang Song, Kamalika Chaudhuri, and Anand~D Sarwate.
\newblock Stochastic gradient descent with differentially private updates.
\newblock In {\em 2013 IEEE Global Conference on Signal and Information
  Processing}, pages 245--248. IEEE, 2013.

\bibitem{ullman}
Jonathan Ullman.
\newblock Cs7880: Rigorous approaches to data privacy, spring 2017.

\bibitem{wang2019subsampled}
Yu-Xiang Wang, Borja Balle, and Shiva~Prasad Kasiviswanathan.
\newblock Subsampled r{\'e}nyi differential privacy and analytical moments
  accountant.
\newblock In {\em The 22nd International Conference on Artificial Intelligence
  and Statistics}, pages 1226--1235. PMLR, 2019.

\bibitem{yousefpour2021opacus}
Ashkan Yousefpour, Igor Shilov, Alexandre Sablayrolles, Davide Testuggine,
  Karthik Prasad, Mani Malek, John Nguyen, Sayan Ghosh, Akash Bharadwaj,
  Jessica Zhao, et~al.
\newblock Opacus: User-friendly differential privacy library in pytorch.
\newblock {\em arXiv preprint arXiv:2109.12298}, 2021.

\bibitem{zhu2022optimal}
Yuqing Zhu, Jinshuo Dong, and Yu-Xiang Wang.
\newblock Optimal accounting of differential privacy via characteristic
  function.
\newblock In {\em International Conference on Artificial Intelligence and
  Statistics}, pages 4782--4817. PMLR, 2022.

\end{thebibliography}

\newpage
\appendix
\noindent
{\Large {\textbf{Supplementary Material}}}
\\
\section{Additional Proofs}
\label{app:proof}
\subsection{Proof of Theorem \ref{th:subshuff}}
\begin{proof}
Theorem 9 of \cite{wang2019subsampled} states that for subsampling rate $\gamma$, a mechanism $\mathcal{M}$ satisfying $(\lambda,\eps(\lambda))$-RDP applied to a subsampled set of data, satisfy $(\lambda,\eps'(\lambda))$, where
\begin{align*}
\eps'(\lambda) 
&\leq \frac{1}{\lambda -1}
\log [1 + \gamma^2 \binom{\lambda}{2} \min\left\{4(e^{\eps_m(2)}-1),  e^{\eps_m(2)}\min\{2,(e^{\eps(\infty)-1})^2\}\right\} \\
&+\sum_{j=3}^{\lambda}\gamma^j \binom{\lambda}{j} e^{(j-1)\eps_m(j)} \min \{2, (e^{\eps(\infty)-1})^j\}
]
\end{align*}
Since $\eps(\infty)$ is unbounded, we can simplify it to the expression given in Theorem \ref{th:subshuff} in a straightforward way.
\end{proof}

\subsection{Proof of Theorem \ref{th:sci}}
\begin{proof}

The Chernoff bound states that 
$\mathbb{P}[X \leq (1-\Delta)\mu] \leq e^{-\Delta^2\mu/2}$ for all $0 < \Delta < 1$.

We split the summation over $k$ of
$$
 \sum_{k=1}^{n}  \binom{n}{k} \gamma^{k} (1-\gamma)^{n-k} e^{(\lambda-1)\eps^{SSG}_{\gamma,k}}
$$
into two parts, those equal or less than $(1-\Delta)n\gamma$, and those equal or larger than $(1-\Delta)n\gamma+1$, assuming that $(1-\Delta)n\gamma$ is an integer.
As Conjecture \ref{cj:mono} tells us that the expression in Equation \ref{eq:main4} is monotonically decreasing, the largest value of $\eps^{\rm SSG}_{\gamma,k}(\lambda)$ for $k \in [1,(1-\Delta)n\gamma]$ is $\eps^{\rm SSG}_{\gamma,1}(\lambda)$.
The first part of the summation can then be bounded by
$$
e^{-\Delta^2n\gamma/2}\eps^{\rm SSG}_{\gamma,1}(\lambda),
$$
using the Chernoff bound.

Similarly, the largest value of $\eps^{\rm SSG}_{\gamma,k}(\lambda)$ for $k \in [(1-\Delta)n\gamma]+1, n$ is $\eps^{\rm SSG}_{\gamma,(1-\Delta)n\gamma]+1}(\lambda)$.
The second part of the summation is then bounded by 
$$
\eps^{\rm SSG}_{\gamma,(1-\Delta)n\gamma+1}(\lambda)
$$
Combining the above two terms gives the desired theorem.
\end{proof}
\section{Algorithms}
\label{app:algo}
\subsection{$\texttt{GetUniqueCount}$: Note on the permutation invariance of Equation \ref{eq:main5}}
It is best to describe the counting factor in our numerical evaluation of Equation \ref{eq:main5} via an example.

Consider again $(x_1+x_2+x_3)^3$. 
Expanding the multinomials, the terms with exponents of the form $x_i^2x_j$ are
$$
3(x_1^2x_2+x_1^2x_3+x_2^2x_1+x_2^2x_3+x_3^2x_1+x_3^2x_2).
$$
Here, the factor 3 comes from the multinomial coefficient $\binom{3}{2,1,0}$.
These terms $x_1^2x_2,x_1^2x_3,x_2^2x_1,x_2^2x_3,x_3^2x_1,x_3^2x_2$ belong to the same subset of the form $x_i^2x_j$, which do not have repetition of exponent with same degrees.
It has $3!= 6$ elements in total.
This subset contributes effectively a factor of $3\times 6=18$ to the expansion in Equation \ref{eq:main5},
$$
18e^{(2^2+1^2)/2\sigma^2},
$$
ignoring factor unrelated to the multinomial coefficients.

Consider the same expansion, but the terms with exponents of the form $x_ix_jx_k$:
$$
6x_1x_2x_3.
$$
Here, the multinomial coefficient is $\binom{3}{1,1,1}=6$.
Since the number of repetition of exponent with the same degree is 3 ($x_i,x_j,x_k$ has the same degree 1), the subset has $3!/(3!)=1$ element.
This subset contributes effectively a factor of $6\times 1=6$ to the expansion in Equation \ref{eq:main5},
$$
6e^{(1^1+1^2)/2\sigma^2},
$$
ignoring factor unrelated to the multinomial coefficients.

This procedure of calculating the contributing coefficients is called $\texttt{GetUniqueCount}$ in Algorithm \ref{alg:num}.

\subsection{Algorithm \ref{alg:num} description}
We describe our algorithm for evaluating Equation \ref{eq:main}.

Given $\lambda$, we find all partition of integers satisfying $k_1,\dotsc,k_n=\lambda$.
We denote the operation by $\texttt{GetPartition}$.
For each of the partition, we obtain the number of counts for each unique $k_i$; let them be $\kappa_1,\dotsc,\kappa_i$ ($i\leq n$), and denote the operation by $\texttt{GetUniqueCount}$.
Subsequently, we calculate the value:
$$
\frac{e^{-\lambda/2\sigma^2}}{n^\lambda}\binom{\lambda}{k_1\dotsc k_n}\frac{n!}{\kappa_1!\dotsc\kappa_i!} e^{\sum_{j=1}^{n}k_j^2},
$$
and make summation over all the partitions.
The algorithm is shown in Algorithm \ref{alg:num}.

}
\end{document}